\documentclass[11pt,a4paper]{article}
\usepackage{amssymb,amsmath, amsfonts,graphicx,multirow,amsthm}
\usepackage{color}
\usepackage{epstopdf}
\usepackage{lscape}
\usepackage{booktabs}
\usepackage{caption}
\usepackage{hyperref}
\usepackage{breqn}
\hypersetup{colorlinks=true, citecolor=blue, linkcolor=red}
\usepackage[round]{natbib}
\usepackage[margin=1in,footskip=0.5in]{geometry}
\numberwithin{equation}{section}

\newtheorem{theorem}{Theorem}[section]

\newtheorem{application}{Application}[section]

\def\correspondingauthor{\footnote{Corresponding author : bhagwatistats@gmail.com (Bhagwati Devi)}}

\begin{document}
	\begin{center}
		{\Large\bf Inference of Half Logistic Geometric Distribution Based on  Generalized Order Statistics}\\
		\vspace{0.5cm}
		
		Neetu Gupta$^{a}$, S. K. Neogy$^{a}$, Qazi J. Azhad$^{b}$\correspondingauthor{}, Bhagwati Devi$^{c}$ \\
		$^{a}$Indian Statistical Institute, Sanswal Marg, New Delhi.\\
		$^{b}$Department of Mathematics, Shiv Nadar Institution of Eminence, Dadri, India\\
		$^{c}$Department of Statistics, Central University of Jharkhand, Jharkhand, India\\
	\end{center}
	\hrule
	\begin{abstract}
		\noindent As the unification of various models of ordered quantities, generalized order statistics act as a simplistic approach introduced in \cite{kamps1995concept}. In this present study, results pertaining to the expressions of marginal and joint moment generating functions from half logistic geometric distribution are presented based on generalized order statistics framework. We also consider the estimation problem of $\theta$ and provides a Bayesian framework. The two widely and popular  methods called Markov chain Monte Carlo and Lindley approximations are used for obtaining the Bayes estimators.The results are derived under symmetric and asymmetric loss functions. Analysis of the special cases of generalized order statistics, \textit{i.e.,} order statistics is also presented. To have an insight into the practical applicability of the proposed results, two real data sets, one from the field of Demography and, other from reliability have been taken for analysis.\\
		
	\end{abstract}
	\section{Introduction}
	Generalized order statistics (\textit{gos}) have been introduced as an extension of ordinary order statistics by using quantile transformation based on the distribution function $F$ (see \cite{kamps1995concept}. Interestingly, in view of the joint density of respective \textit{gos}, various sub models for ordered data, such as order statistics, record values and progressively Type-II censored order statistics are included in this set-up. Various researchers have studied \textit{gos} in various directions such as recurrence relations of moments, estimation of parameters, distributional properties of different distributions, etc. Being a general structure to variety of models, it is widely popular among researchers. The \textit{gos} provide a general framework to deal with increasingly ordered random variables. \par
	Assuming $X_{1} ,X_{2} ,...$ be a sequence of independent identically distributed (\textit{iid}) random variables with distribution function (\textit{df}) $F\left(x\right)$ and probability density function (\textit{pdf}) $f(x)$. Let $k\ge 1,n\ge 2,n\in N,\tilde{m}=\left(m_{1} ,m_{2} ,...,m_{n-1} \right)\in \Re ^{n-1} ,M_{r} =\sum_{j=r}^{n-1}m_{j}$ such that $\gamma _{r\,}=k+n-r+M_{r} >0$ $\forall \, \, \, \, \, r\in \left\{1,2,...,n-1\right\}$. Then $X\left(r,n,\tilde{m},k\right)$, $r=1,2,\ldots,n$, is called \textit{gos} based on $F\left(x\right)$, if their  joint \textit{pdf} is of the form
	\begin{eqnarray}\label{eq1}
		k\left(\prod _{j=1}^{n-1}\gamma _{j}  \right)\left(\prod _{i=1}^{n-1}\left[1-F\left(x_{i} \right)\right]^{m_{i} } f\left(x_{i} \right) \right)\left[1-F\left(x_{n} \right)\right]^{k-1} f\left(x_{n} \right)\nonumber,\, \, \, \, \, \, \, \, \, \, \, \, \, \, \, \, \,\\
		F^{-1} \left(0\right)<x_{1} \le ....\le x_{n} <F^{-1} \left(1\right).
	\end{eqnarray}
	Considering appropriate values of model parameters  $m_{i}$, $k$ and $\gamma _{i}$ the model of \textit{gos} reduces to various important models e.g., when $m_{1} =m_{2} =...=m_{n-1} =0,$ $k=1$, $\gamma _{i} =\left(1+n-i\right)$, this model reduces to order statistics and for $m_{1} =m_{2} =...=m_{n-1} =-1$, $\gamma _{i} =k$,  $k\in N,$ this model reduces to $k^{th}$ upper record values.
	Here, we consider the case when $m_{i} =m_{j} =m$. Then the density function of $r^{th} $ \textit{gos} $X\left(r,n,m,k\right)$ is given by
	\begin{equation} \label{gos.cdf}
		f_{X\left(r,n,m,k\right)} =\frac{C_{r-1} }{\left(r-1\right)!} \left[\bar{F}\left(x\right)\right]^{\gamma _{r} -1} f\left(x\right)g_{m}^{r-1} \left[F\left(x\right)\right].     
	\end{equation} 
	The joint \textit{pdf} of $X\left(r,n,m,k\right)$ and $X\left(s,n,m,k\right)$, $1\le r<s\le n$ is
	\begin{align}  \label{jointpdf.gos}
		f_{X\left(r,n,m,k\right)X\left(s,n,m,k\right)} &=\frac{C_{s-1} }{\left(r-1\right)!\left(s-r-1\right)!} \left[\bar{F}\left(x\right)\right]^{m} f\left(x\right)g_{m}^{r-1} \left[F\left(x\right)\right] \nonumber \\
		&\times \left[h_{m} F\left(y\right)-h_{m} F\left(x\right)\right]^{s-r-1}\left[\bar{F}\left(y\right)\right]^{\gamma _{s} -1} f\left(y\right)\, \, , 
	\end{align}
	where,
	\[\bar{F}\left(x\right)=1-F\left(x\right)  , C_{r-1} =\prod _{i=1}^{r}\gamma _{i}   ,\] 
	\[h_{m} \left(x\right)=\left\{\begin{array}{c} {-\frac{1}{m+1} \left(1-x\right)^{m+1}, \qquad m\ne -1} \\ {-\ln \left(1-x\right),\qquad \qquad m=-1} \end{array}\right. \] 
	and
	\[g_{m} \left(x\right)=h_{m} \left(x\right)-h_{m} \left(0\right),   x\in \left[0,1\right).\] 
	
	Since its emergence, numerous researchers have included the concept of \textit{gos} in their work such as recurrence relation of moments, estimation of parameters, distributional properties for various distributions etc. (see \cite{aboeleneen2010inference}, \cite{burkschat2009linear}, \cite{gupta2019inference} and \cite{arshad2021bayesian}  ) This article deals with the moment generating function (mgf) for calculating higher-order moments and Bayesian estimation of parameter of HLG distribution. Moments and recurrence relations grabbed a lot of attention in literature. Recurrence relations are exciting and useful for the computation of higher-order moments. They show a great ability in reducing the amount of time and number of operations necessary to get a general form of the function under the consideration of order. Several authors have derived expressions for moments based on \textit{gos} and their special cases for various continuous distributions. \cite{gupta2019inference} presented relations of moments of WGED based on \textit{gos} setup, \cite{akhter2020moments}  derived the moments and applied the results to the lifetimes of the coherent systems of order statistics from the standard two-sided power distribution. For more detailed reading readers are advised to see \cite{singh2021exact}, \cite{alharbi2021moment}, \cite{khatoon2021moments}.\par
	The estimation of parameters based on \textit{gos} for various lifetime models has been investigated well hitherto. For example, \cite{safi2013generalized} obtained the estimates of the unknown parameters for Kumaraswamy distribution based on \textit{gos} by assuming the maximum likelihood method. In \cite{kim2014bayesian}, Bayesian estimators and highest posterior density credible intervals from Rayleigh distribution of the scale parameter from \textit{gos} setup are obtained. In \cite{azhad2020estimation}, the problem of estimation of common location parameter of several heterogeneous exponential populations is discussed, and various estimators of common location parameter are derived under \textit{gos} setup.\par
	Because of their great applications in almost all the parametric statistical techniques, the statistical distributions play a significant role in inference, survival analysis, lifetime data modeling, and reliability. The fitting of the data by statistical modeling is an essential tool for analyzing the lifetime data. Today's new breed of statisticians is a sleuth who uses modern analytics to find patterns in vast oceans of information. Almost every area generates data that can be mined for valuable patterns. Sometimes, due to some limitations, pre-existing models fail to fit the lifetime data. Several lifetime distributions have been developed in the literature for this purpose. The past decade has witnessed tremendous growth in the development of statistical models. It has seen that the widely used lifetime models usually have a limited range of behaviors. Such types of distributions cannot give a better fit to model all the practical situations. Recently, several authors have developed different families of statistical models by applying various techniques.
	The logistic distribution was first used by \cite{verhulst1838notice,verhulst1845loi} for the growth models with modeling lifetime data, plays a vital role in survival analysis as a parametric model.   \cite{gumbel1944ranges} introduced the genesis of this distribution. It has many applications; for example, mortality from cancer following diagnosis or treatment, steam flow and precipitation (in hydrology), and the distribution of wealth or income (in economics). Having remarkable properties and applications in various areas, some of the researchers argued of the left-hand limit of this distribution approach to negative infinity. They found it inappropriate in many situations of modeling lifetime data, which resulted in developing new models.
	The half logistic distribution has received a considerable attention in the statistical literature. It was obtained by folding the logistic distribution and considered by \cite{balakrishnan1985order} in the context of order statistics. 
	The properties of analyzing lifetime data and the increasing failure rate makes it attention grabbing among researchers. Recently, a lot of  work has been done on the  half logistic distribution; literature consists of \cite{balakrishnan2013recurrence}, \cite{jodra2014logarithmic} listed the properties of order statistics of half logistic families of distribution. In recent work by \cite{alotaibi2020bivariate}, \cite{almarashi2020statistical}, \cite{jeon2020estimation}, \cite{xavier2020study}, and \cite{zhang2020parameter} revealed enormous study on the HL distribution. A random variable $X$ follows half logistic distribution if its pdf is given as
	\begin{equation}
		f(x)=\dfrac{2e^{-x}}{(1+e^{-x})^2}, x>0.
	\end{equation}
	and the corresponding df is given as
	\begin{equation}
		F(x)=\dfrac{(1-e^{-x})}{(1+e^{-x})}, x>0.
	\end{equation}
	Despite having tremendous properties in lifetime modeling, half logistic distribution fails to perform in engineering problems like: to know the exact number of components fail. In these types of situations, \cite{liu2020recurrence}  used a discrete distribution with Half logistic distribution to overcome the drawback of logistic distribution, and introduced a new model. The HLG random variable X possess \textit{df} as
	\begin{equation}\label{cdf.hlg}
		F(x)=\dfrac{\theta(1-e^{-x})}{(\theta+(2-\theta)e^{-x})}, x>0,0<\theta<1.
	\end{equation}
	This new distribution is obtained by mixing half logistic distribution with a geometric distribution. The Half logistic geometric (HLG) distribution has an increasing failure rate for all values of shape parameters. Having one extra parameter makes this distribution very flexible and it provides more flexibility than Half logistic distribution and very useful in modeling data.\\
	The pdf is 
	\begin{equation}\label{pdf.hlg}
		f(x)=\dfrac{2\theta e^{-x}}{(\theta+(2-\theta)e^{-x})^2}, x>0,0<\theta<1.
	\end{equation}
	It can be noticed that 
	\begin{equation}\label{rel.pc}
		f(x)=\bar{F}\left(x\right)-\left(\frac{2-\theta}{2}\right){\bar{F}\left(x\right)}^2.
	\end{equation}
	where $\bar{F}\left(x\right)=1-{{F} \left(x\right)}$.\\
	The relation in \eqref{rel.pc} will further be utilized to establish the recurrence relations for single and product moments of \textit{gos} for HLG distribution.
	In \cite{liu2020recurrence}, the authors have discussed the recursive form for moments of order statistics. To extend the work of \cite{liu2020recurrence}, we have discussed the closed form expression as well as recursive relations for marginal and joint moments for  \textit{gos} based on HLG distribution and also reduce them to single and product moments of order statistics for the same distribution.
	 In this paper, the study clearly provides a more general view of HLG distribution in terms of practical applicability. \\

	The organization of present article is as follows: In Sections [\ref{sec2}] and [\ref{sec:joint}], explicit expressions for marginal and joint mgfs based on \textit{gos} from HLG distribution are derived and after that the results are further reduced to the sub model of \textit{gos}. In Section [\ref{sec:bayes}], approximate Bayes estimators are obtained with the aid of Lindley approximation and Markov chain Monte Carlo method. In Section [\ref{sec:simu}], a simulation study is reported for the sub model of \textit{gos} i.e., order statistics with detailed analysis. In section [\ref{sec:real}], two real data applications are given from the field of demography and reliability. 
	
\section{Marginal moments}\label{sec2}
	In this section, an expression of mgf of marginal moments of HLG distribution is obtained under the setup of \textit{gos}. Let $M_{(r,n,m,k)}(t)$ $=E(e^{tX_{r,n,m,k}})$ be the mgf for $r^{th}$ \textit{gos}.
	\begin{theorem}\label{theo1}
		For a continuous non-negative random variable X distributed as (1.8), with real $m,k$,  $m\ge-1$, $k\ge1$ 
	\end{theorem}
	\begin{eqnarray} \label{mgf.gos}
		M_{(r,n,m,k)}(t)=\dfrac{c_{r-1}}{(r-1)!(m+1)^{r-1}}\sum_{u=0}^{r-1}(-1)^{u}\binom{r-1}{u}\left(\frac{2}{2-\theta}\right)^{\gamma_{r-u}}\left(\frac{\theta}{2-\theta}\right)^{-t} \nonumber\\
		\, \, \, \, \times B\left(1-\frac{\theta}{2},\gamma_{r-u}-t,1+t\right),\,\,\,\,m\ne-1.\,\,\,\,\,\,\,\,
	\end{eqnarray}
	\begin{proof}
		By using \eqref{gos.cdf} the \textit{mgf} of \textit{gos} can be seen as
		\begin{equation} 
			M_{(r,n,m,k)}(t)=\frac{c_{r-1}}{(r-1)!}\int_{0}^\infty e^{tx}\left[\bar{F}\left(x\right)\right]^{\gamma _{r} -1} g_{m}^{r-1} \left(F\left(x\right)\right)f\left(x\right)dx.
		\end{equation}
		Expanding $g_{m}^{r-1}[F(x)]$ binomially in (2.3), we get
		\begin{equation} 
			M_{r:n,m,k} =\frac{C_{r-1} }{\left(r-1\right)!(m+1)^{r-1} }  \sum_{u=0}^{r-1}\binom{r-1}{u}(-1)^{u} \int _{0}^{\infty }e^{tx} \left[\bar{F}\left(x\right)\right]^{\gamma _{r-u} -1} f\left(x\right)dx  .  
		\end{equation}
		Making appropriate substitution by considering \eqref{cdf.hlg} and \eqref{pdf.hlg}, and after simplification, we get \eqref{mgf.gos}. To show the existence of $\mu_{(r,n,m,k)}=E\left(X_{r,n,n,k}\right)$, we proceed as
		\begin{equation}
			\mu_{(r,n,m,k)}=\frac{d}{dt}M_{(r,n,m,k)}(t)\Bigg|_{t=0}\nonumber
		\end{equation}
		\begin{eqnarray}\label{momet1.gos}
			\mu_{(r,n,m,k)}=\dfrac{C_{r-1}}{(r-1)!(m+1)^{r-1}}\sum_{u=0}^{r-1}(-1)^{u}\binom{r-1}{u}\left(\dfrac{2}{2-\theta}\right) \left[ln\left(\dfrac{\theta}{2-\theta}\right)-\mathcal{A}^{'}\right],
		\end{eqnarray} 
		where,
		\begin{equation}
			\mathcal{A}^{'}={\frac{d}{dt}}B\left[1-\frac{1}{\theta},\gamma_{r-u}-t,1+t\right]\Bigg|_{t=0}\nonumber
		\end{equation}
		Similarly, the second moment of \textit{gos} from HLG distribution can be obtained as
		\begin{equation}
			\mu_{(r,n,m,k)}^{2}=\frac{d^{2}}{dt^{2}}M_{(r,n,m,k)}(t)\Bigg|_{t=0}\nonumber
		\end{equation}
		\begin{eqnarray}\label{moment2.gos}
			\mu_{(r,n,m,k)}^{2}=\dfrac{C_{r-1}}{(r-1)!(m+1)^{r-1}}\sum_{u=0}^{r-1}(-1)^{u}\binom{r-1}{u}\left(\dfrac{2}{2-\theta}\right)\nonumber\\
			\times\left[ln\left(\dfrac{\theta}{2-\theta}\right)^{2}B\left(1-\frac{1}{\theta},\gamma_{r-u}-t,1+t\right)-2ln\left(\dfrac{\theta}{2-\theta}\right)\mathcal{A}^{'}-\mathcal{A}^{''}\right],
		\end{eqnarray}
		where, 
		\begin{equation}
			\mathcal{A}^{''}={\frac{d^{2}}{dt^{2}}}B\left[1-\frac{1}{\theta},\gamma_{r-u}-t,1+t\right]\Bigg|_{t=0}\nonumber
		\end{equation}
		Proceeding in similar way, \eqref{mgf.gos} can produce $p$-th moment of HLG distribution based on \textit{gos}. 
		\end{proof}
\noindent The results provided in the Theorem [\ref{theo1}] can easily by reduced to the submodels of \textit{gos}. In this direction, we have reduced the results to order statistics by considering $m=0$ and $k=1$ in \eqref{mgf.gos} and obtained marginal mgf for order statistics from HLG distribution as
	\begin{eqnarray}
	M_{r:n}(t)=C_{r,n}\sum_{u=0}^{r-1}(-1)^{u}\binom{r-1}{u}{\left(\frac{2}{2-\theta}\right)}^{(n-r+u+1)}{\left(\frac{\theta}{2-\theta}\right)}^{-t} B{\left(1-\frac{\theta}{2},n-r+u+1-t,1+t\right)}.\nonumber
\end{eqnarray}
where $C_{r,n}=\dfrac{n!}{(r-1)!(n-r)!}$. For $r=1$ and $n=1$, result coincides with \cite{liu2020recurrence} as
 			\begin{equation}
			M(t)={\left(\frac{2}{2-\theta}\right)}{\left(\frac{\theta}{2-\theta}\right)}^{-t}B{\left(1-\frac{\theta}{2},1-t,1+t\right)}.\nonumber\\
		\end{equation}
	
\noindent Also, putting $m=0$ and $k=1$ in \eqref{momet1.gos} and \eqref{moment2.gos}, we obtain first and second moments of order statistics of HLG distribution based on \textit{gos} setup. For $r=1$ and $n=1$, the results agrees with \cite{liu2020recurrence}. The numerical results of mean for order statistics for various configurations are obtained by \cite{liu2020recurrence}.
		
	\begin{theorem}\label{theo2}
		For integers $r\ge1$, $p\ge1$ and under the assumption of Theorem \ref{theo1}, the given recurrence relation is satisfied.
		\begin{eqnarray}\label{reccr.mgf}
			M_{r,n,m,k}^{p}(t)=\dfrac{{\gamma_{r}}}{(p+1)}\left[M_{r,n,m,k}^{p+1}(t)-M_{r-1,n,m,k}^{p+1}(t)\right]-\left(\dfrac{2}{2-\theta}\right)\dfrac{\gamma_{r}^{n+1,k-m}}{(p+1)}\nonumber \\ \,\,\,\,\,\, \times \left[M_{r,n+1,m,k-m}^{p+1}(t)-M_{r-1,n+1,m,k-m}^{p+1}(t)\right]- \dfrac{t}{(p+1)}M_{r,n,m,k}^{p+1}(t),\,\,\,\,\,\,\, m\ne-1\
		\end{eqnarray}
	\end{theorem}
	\begin{proof}
		In view of \eqref{gos.cdf}, we have 
		\begin{equation}\label{2.7}
			M_{r,n,m,k}(t)=\dfrac{c_{r-1}}{(r-1)!}\int_{0}^\infty{e^{tx}}\left[\bar{F}\left(x\right)\right]^{\gamma _{r} -1}g_{m}^{r-1}\left(F(x)\right) f\left(x\right)dx.
		\end{equation}
		Making use of \eqref{rel.pc} in \eqref{2.7}, we have
		\begin{align}\label{2.8}
			M_{r,n,m,k}(t)&=\dfrac{c_{r-1}}{(r-1)!}\int_{0}^\infty{e^{tx}}\left[\bar{F}\left(x\right)\right]^{\gamma _{r}}g_{m}^{r-1}\left(F(x)\right)dx
			-{\left(\dfrac{2-\theta}{2}\right)}\dfrac{c_{r-1}}{(r-1)!}\nonumber\\
			&\times\int_{0}^\infty{e^{tx}}\left[\bar{F}\left(x\right)\right]^{\gamma _{r} +1}g_{m}^{r-1}\left(F(x)\right)dx.\nonumber \\
			&=Z(x)-{\left(\dfrac{2-\theta}{2}\right)}Z^{*}(x)
		\end{align}
		where,
		\begin{equation}
			Z(x)=\dfrac{c_{r-1}}{(r-1)!}\int_{0}^\infty{e^{tx}}\left[\bar{F}\left(x\right)\right]^{\gamma _{r}}g_{m}^{r-1}(F(x)dx\nonumber
		\end{equation}
		and 
		\begin{equation}
			Z^{*}(x)=\dfrac{c_{r-1}}{(r-1)! }\int_{0}^\infty{e^{tx}}\left[\bar{F}\left(x\right)\right]^{\gamma _{r}+1}g_{m}^{r-1}F(x)dx\nonumber
		\end{equation}
		After integrating $Z(x)$ by parts, we obtain
		
		\begin{equation}
			Z(x)=\dfrac{\gamma_{r}}{t}\left[M_{r,n,m,k}(t)-M_{r-1,n,m,k}(t)\right]\nonumber
		\end{equation}
		In similar way, integrating $Z^{*}(x)$, we obtain
		\begin{equation}
			Z^{*}(x)=\dfrac{\gamma_{r}^{n+1,k-m}C^{*}}{t}\left[M_{r,n+1,m,k-m}(t)-M_{r-1,n+1,m,k-m}(t)\right]\nonumber
		\end{equation}
		where $C^{*}=\dfrac{c_{r-1}}{c_{r-1}^{n+1,k-m}}.$\\
		Substituting $Z(x)$ and $Z^{*}(x)$ in \eqref{2.8}, we obtain\\
		\begin{eqnarray}\label{mgf.reccr}
			tM_{r,n,m,k}(t)={\gamma_{r}}\left[M_{r,n,m,k}(t)-M_{r-1,n,m,k}(t)\right]-\left(\dfrac{2-\theta}{2}\right){\gamma_{r}^{n+1,k-m}C^{*}}\nonumber \\ \,\,\,\,\,\, \times \left[M_{r,n+1,m,k-m}{(t)}-M_{r-1,n+1,m,k-m}(t)\right].
		\end{eqnarray}
		Differentiating \eqref{mgf.reccr} $(p+1)$ times w.r.t. $t$, we obtain \eqref{reccr.mgf}.\\
		We get the recurrence relation for single moments from HLG distribution based on $\textit{gos}$ by taking $t=0$ in \eqref{reccr.mgf}, as
		\begin{eqnarray}\label{2.10}
			\mu_{r,n+1,m,k-m}^{p+1}=-{\left(\dfrac{2}{2-\theta}\right)}\dfrac{\gamma_{r}}{C^{*}\gamma_{r}^{k-m,n+1}}\left[\dfrac{(p+1)}{\gamma_{r}}\mu_{r,n,m,k}^{p}-\mu_{r,n,m,k}^{p+1}+\mu_{r-1,n,m,k}^{p+1}\right]\nonumber\\
			+\mu_{r-1,n+1,m,k-m}^{p+1},\,\,\,\,\,\,\, m\ne-1.\
		\end{eqnarray}
	\end{proof}
\noindent The results provided in the Theorem [\ref{theo2}] are reduced to order statistics by putting $m=0$ and $k=1$ in \eqref{mgf.reccr}. Thus, recurrence relation for marginal mgf for order statistics from HLG distribution is 
		\begin{eqnarray}
			tM_{r:n}(t)={(n-r+1)}\left[M_{r:n}(t)-M_{r-1:n}(t)\right]-\left(\dfrac{2-\theta}{2}\right)\dfrac{\left(n-r+2\right)\left(n-r+1\right)}{(n+1)}\nonumber \\ \,\,\,\,\,\,  \times\left[M_{r:n+1}{(t)}-M_{r-1:n+1}(t)\right].
		\end{eqnarray}
\noindent	 For $m=0$ and $k=1$ in \eqref{2.10}, we obtain the recurrence relation for order statistics from HLG distribution which coincides with \cite{liu2020recurrence} as 
		\begin{eqnarray}
			\mu_{r:n+1}^{p+1}=-{\left(\dfrac{2}{2-\theta}\right)}\dfrac{(n+1)}{(n-r+2)}\left[\dfrac{(p+1)}{(n-r+1)}\mu_{r:n}^{p}-\mu_{r:n}^{p+1}+\mu_{r-1,n}^{p+1}\right]+\mu_{r-1:n+1}^{p+1}.\nonumber
		\end{eqnarray}

\section{Joint Moments} \label{sec:joint}
	\begin{theorem}\label{mgf.theo}
		For $1\leq r<s\leq n$, $n \in{N}$, $m \in {R}$, and a fixed positive integer $k\ge 1$	
	\end{theorem}
	
	\begin{eqnarray} \label{mgf.joint}
		M_{r,s,n,m,k}(t_{1},t_{2}) =\dfrac{C_{s-1} }{\left(r-1\right)!\left(s-r-1\right)!(m+1)^{s-2}}\sum_{c=0}^{\infty}\sum_{u=0}^{r-1}\sum_{v=0}^{s-r-1}\binom{s-r-1}{v}\nonumber\\
		\times\binom{r-1}{u}(-1)^{u+v}\left(\frac{2}{2-\theta}\right)^{\gamma_{r-u}+t_{1}}\left(\frac{\theta}{2-\theta}\right)^{-(t_{1}+t_{2})}\dfrac{\left(\gamma_{s-v}+1\right)_{c}(1)_{c}}{\left(\gamma_{s-v}-t_{2}+1\right)_{c}\left(\gamma_{s-v}-t_{2}\right){c}!}\nonumber\\
		\times B\left(1-\frac{\theta}{2},\gamma_{r-u}-t_{2}+c,t_{1}+t_{2}+2\right),\text{for} \, m\neq {-1}.
	\end{eqnarray}
	\begin{proof}
		Making use of \eqref{jointpdf.gos}, we obtain the joint mgf from \textit{gos} as
		\begin{eqnarray}  
			M_{\left(r,s,n,m,k\right)}({t_{1},t_{2}}) =\frac{C_{s-1} }{\left(r-1\right)!\left(s-r-1\right)!}\int_{0}^{\infty}\int_{x}^{\infty} e^{t_{1}x+t_{2}y  }\left[\bar{F}\left(x\right)\right]^{m} f\left(x\right)g_{m}^{r-1} \left(F\left(x\right)\right) \nonumber \nonumber \\
			\times \left[h_{m} F\left(y\right)-h_{m} F\left(x\right)\right]^{s-r-1}\left[\bar{F}\left(y\right)\right]^{\gamma _{s} -1} f\left(y\right)dydx.\nonumber
		\end{eqnarray}
		Simplifying above expression by using binomial expansion of $g_{m}^{r-1}   \left(F\left(x\right)\right)$ \text{and} $\left[h_{m} F\left(y\right)-h_{m} F\left(x\right)\right]^{s-r-1}$, we obtain
		\begin{eqnarray}
			M_{\left(r,s,n,m,k\right)}({t_{1},t_{2}}) =\frac{C_{s-1} }{\left(r-1\right)!\left(s-r-1\right)!(m+1)^s-2}\sum_{c=0}^{\infty}\sum_{u=0}^{r-1}\sum_{v=0}^{s-r-1}\binom{s-r-1}{v}\binom{r-1}{u}\nonumber\\
			\times\int_{0}^{\infty} e^{t_{1}x}\left[\bar{F}\left(x\right)\right]^{(s-r+u-v)(m+1)-1}f(x)I(x)dx.
		\end{eqnarray}
		where,
		\begin{equation}
			I(x)=\int_{x}^{\infty}e^{t_{2}y}\left[\bar{F}\left(y\right)\right]^{\gamma _{s-v} -1} f\left(y\right)dy.\nonumber
		\end{equation}
		In view of \eqref{cdf.hlg} and \eqref{pdf.hlg} and making appropriate substitution in $I(x)$, we obtain
		\begin{eqnarray}
			I(x)=\theta^{-t_{2}}2^{t_{2}}\left(\dfrac{2}{2-\theta}\right)^{\gamma_{s-v}-t_{2}}\int_{0}^{\left(\dfrac{2-\theta}{2}\right)\bar{F}\left(x\right)}
			{z^{\gamma_{s-v}-t_{2}-1}}(1-z)^{t_{2}}dz.
		\end{eqnarray}
		In view of the following integral used by \cite{dutka1981incomplete},
		\begin{equation}
			\int_{0}^{x}u^{a-1}(1-u)^{b-1}du=\dfrac{x^{a}(1-x)^b}{a}{}_{2}F_{1}(a+b,1;a+1;x)
		\end{equation}
		Here, $ {}_{2}F_{1}\left(\alpha,\beta,\gamma,z\right)$ is Gauss Hypergeometric function defined as
		\begin{equation}
			{}_{2}F_{1}\left(\alpha,\beta,\gamma,x\right)=\sum_{z=0}^{\infty}\dfrac{(\alpha)_{z}(\beta)_{z}}{(\gamma)_{z}}\frac{x^{z}}{z!}.\nonumber
		\end{equation}
		simplifying $I(x)$ by using (3.4) and we reduce (3.2) as
		\begin{eqnarray}  
			M_{\left(r,s,n,m,k\right)}({t_{1},t_{2}})= \dfrac{C_{s-1} }{\left(r-1\right)!\left(s-r-1\right)!(m+1)^{s-2}}\sum_{c=0}^{\infty}\sum_{u=0}^{r-1}\sum_{v=0}^{s-r-1}\binom{s-r-1}{v}\nonumber\\
			\times\binom{r-1}{u}(-1)^{u+v}\dfrac{\theta^{-t_{2}}2^{t_{2}}}{\left(\gamma_{s-v}-t_{2}\right)}\int_{0}^{\infty}e^{t_{1}x}\left[\bar{F}\left(x\right)\right]^{\gamma_{s-v}-t_{2}+(s-r+u-v)(m+1)-1}\nonumber\\
			\left[1-\left(\frac{2-\theta}{2}\right)\bar{F}\left(x\right)\right]^{t_{2}+1}{}_{2}F_{1}\left(\gamma_{s-v}+1,1;\gamma_{s-v}-t_{2}+1;\left(\dfrac{2-\theta}{2}\right)\bar{F}\left(x\right)\right)
		\end{eqnarray}
		Solving (3.5) by expanding ${}_{2}F_{1}\left(\gamma_{s-v}+1,1;\gamma_{s-v}-t_{2}+1;\left(\dfrac{2-\theta}{2}\right)\bar{F}\left(x\right)\right)$and simplifying the resultant, we obtain \eqref{mgf.joint}.
	\end{proof}
\noindent For $m=0$ and $k=1$ in \eqref{mgf.joint}, the expression for joint mgf for order statistics from HLG distribution is obtained as
		\begin{eqnarray} 
			M_{r,s:n}(t_{1},t_{2}) =C_{r,s,n}\sum_{c=0}^{\infty}\sum_{u=0}^{r-1}\sum_{v=0}^{s-r-1}\binom{s-r-1}{v}
			\binom{r-1}{u}(-1)^{u+v}\left(\frac{2}{2-\theta}\right)^{n-r+u+1-t_{1}}\nonumber\\
			\times\left(\frac{\theta}{2-\theta}\right)^{-t_{1}-t_{2}}\dfrac{\left({n-s+v+2}\right)_{c}(1)_{c}}{\left(n-s+v-t_{2}+2\right)_{c}\left(n-s+v-t_{2}+1\right){c}!}\nonumber\\ \times B\left(1-\frac{\theta}{2},n-r+u+c-(t_{1}+t_{2}),(t_{1}+t_{2}+2)\right)\nonumber
		\end{eqnarray}

	\begin{theorem}\label{theo3.2}
		For the distribution given in (1.8) and under the assumptions of Theorem \ref{mgf.theo}, the following recurrence relation is satisfied
		\begin{eqnarray}
			M_{r,s,n,m,k}^{p,q}(t_{1},t_{2})=\dfrac{\gamma_{s}}{(q+1)}\left[M_{r,s,n,m,k}^{p,q+1}(t_{1},t_{2})-M_{r,s-1,n,m,k}^{p,q+1}(t_{1},t_{2})\right]-\left(\dfrac{2-\theta}{2}\right)\dfrac{\gamma_{s}^{k-m,n+1}C^{*}}{(q+1)}\nonumber \\ \,\,\,\,\,\,
			\times\left[M_{r,s,n+1,m,k-m}^{p,q+1}(t_{1},t_{2})-M_{r,s-1,n+1,m,k-m}^{p,q+1}(t_{1},t_{2})\right]- \dfrac{t_{2}}{(q+1)}M_{r,n,m,k}^{p,q+1}(t_{1},t_{2}),\,\, m\ne-1\,\,\
		\end{eqnarray}
	\end{theorem}
	\begin{proof}
		In view of (1.3), the joint \textit{mgf} based on  \textit{gos} is
		\begin{eqnarray}
			M_{\left(r,s,n,m,k\right)}({t_{1},t_{2}}) =\frac{C_{s-1} }{\left(r-1\right)!\left(s-r-1\right)!}\int_{0}^{\infty} e^{t_{1}x }\left[\bar{F}\left(x\right)\right]^{m} f\left(x\right)\nonumber \\
			\times g_{m}^{r-1} \left(F\left(x\right)\right)I(x)dx.
		\end{eqnarray}
		\begin{eqnarray}
			I(x)=\int_{x}^{\infty}e^{t_{2}y}\left[h_{m} F\left(y\right)-h_{m} F\left(x\right)\right]^{s-r-1}\left[\bar{F}\left(y\right)\right]^{\gamma _{s} -1} f\left(y\right)dy\nonumber
		\end{eqnarray}
		or, making use of (1.9) in $I(x)$,
		\begin{equation}
			I(x)=W(x)-{\left(\dfrac{2-\theta}{2}\right)}W^{*}(x)
		\end{equation}
		where,
		\begin{eqnarray}
			W(x)= \int_{x}^{\infty}e^{t_{2}y}\left[h_{m} F\left(y\right)-h_{m} F\left(x\right)\right]^{s-r-1}\left[\bar{F}\left(y\right)\right]^{\gamma _{s}}dy\nonumber
		\end{eqnarray}
		Integrating by parts , we have $s\ge r+1$ 
		\begin{eqnarray}
			W(x)= \dfrac{\gamma_{s}}{t_{2}}\int_{x}^{\infty}e^{t_{2}y}\left[h_{m} F\left(y\right)-h_{m} F\left(x\right)\right]^{s-r-1}\left[\bar{F}\left(y\right)\right]^{\gamma _{s}-1}f(y)dy\nonumber\\
			-\dfrac{(s-r-1)}{t_{2}}\int_{x}^{\infty}e^{t_{2}y}\left[h_{m} F\left(y\right)-h_{m} F\left(x\right)\right]^{s-r-2}\left[\bar{F}\left(y\right)\right]^{\gamma _{s}-1}f(y)dy\nonumber
		\end{eqnarray}
		Proceeding in the same way for $W^{*}(x)$
		\begin{eqnarray}
			W^{*}(x)= \dfrac{\gamma_{s}+1}{t_{2}}\int_{x}^{\infty}e^{t_{2}y}\left[h_{m} F\left(y\right)-h_{m} F\left(x\right)\right]^{s-r-1}\left[\bar{F}\left(y\right)\right]^{\gamma _{s}}f(y)dy\nonumber\\
			-\dfrac{(s-r-1)}{t_{2}}\int_{x}^{\infty}e^{t_{2}y}\left[h_{m} F\left(y\right)-h_{m} F\left(x\right)\right]^{s-r-2}\left[\bar{F}\left(y\right)\right]^{\gamma _{s}+m+1}f(y)dy\nonumber
		\end{eqnarray}
		substituting $W(x)$ and $W^{*}(x)$ in (3.8) and then substitute the $I(x)$ in (3.7), we get after simplification 
		\begin{eqnarray}
			t_{2}M_{r,s,n,m,k}(t_{1},t_{2})={\gamma_{s}}\left[M_{r,s,n,m,k}(t_{1},t_{2})-M_{r,s-1,n,m,k}(t_{1},t_{2})\right]\nonumber\\
			-\left(\dfrac{2-\theta}{2}\right){\left(\gamma_{s}+1\right)C^{*}}\left[M_{r,s,n+1,m,k-m}(t_{1},t_{2})-M_{r,s-1,n+1,m,k-m}(t_{1},t_{2})\right]
		\end{eqnarray}
		Differentiating (3.9) $p$ times w.r.t. $t_{1}$ and $(q+1)$ times w.r.t. $t_{2}$, we obtain (3.6).
		We proceed to obtain the recurrence relation for joint moments of \textit{gos} from HLG distribution by setting $t_{1}$, $t_{2}=0$ in (3.6) as
		\begin{eqnarray}
			\mu_{r,n,m,k}^{p,q}=\dfrac{\gamma_{s}}{(q+1)}\left[\mu_{r,s,n,m,k}^{p,q+1}-\mu_{r,s-1,n,m,k}^{p,q+1}\right]\nonumber \\
		-\left(\dfrac{2-\theta}{2}\right)\dfrac{(\gamma_{s}+1)C^{*}}{(q+1)}\left[\mu_{r,s,n+1,m,k-m}^{p,q+1}-\mu_{r,s-1,n+1,m,k-m}^{p,q+1}\right]
		\end{eqnarray}
	\end{proof}
\noindent The results of Theorem [\ref{theo3.2}] are reduced for order statistics by considering $m=0$ and $k=1$ in (3.9). Thus, we get recurrence relation for joint \textit{mgf} as
			\begin{eqnarray}
			t_{2}M_{r,s:n}(t_{1},t_{2})={(n-s+1)}\left[M_{r,s:n}(t_{1},t_{2})-M_{r,s-1:n}(t_{1},t_{2})\right]-\left(\dfrac{2-\theta}{2}\right)\nonumber\\
			\times\dfrac{\left(n-s+2\right)\left(n-s+1\right)}{(n+1)}\left[M_{r,s:n+1}(t_{1},t_{2})-M_{r,s-1:n+1}(t_{1},t_{2})\right]
		\end{eqnarray}
\noindent After putting $p=1$ and $q=0$ in (3.10) and rearranging the terms, the result agrees with \cite{liu2020recurrence}. The numerical results of covariances and variances for order statistics for various configurations are obtained by \cite{liu2020recurrence}.
\section{Preliminary setup for Bayesian estimation}\label{sec:bayes}
	Let $X(1,n,m,k),X(2,n,m,k),\ldots,X(n,n,m,k)$
	be the $n$ $gos$ and $\boldsymbol{x}=(x_1,x_2,\ldots,x_n)$ is the observed value of $\boldsymbol{X}=(X_1,X_2,\ldots,X_n)$. In view of \eqref{eq1}, \eqref{cdf.hlg} and \eqref{pdf.hlg}, the likelihood function can be given as
	\begin{equation}\label{lik.eq}
		L\left(\theta|x\right)=k\left(\prod_{j=1}^{n-1}\gamma_{j}\right)\left(\prod_{i=1}^{n-1}\dfrac{2^{m+1}e^{-x_{i}(m+1)}}{\left(\theta+(2-\theta)e^{-x_{i}}\right)^{m+2}}\right)\left(\dfrac{2^{k}\theta e^{-x_{n}(k)}}{\left(\theta+(2-\theta)e^{-x_{n}}\right)^{k+1}}\right),\,\,\,  0<\theta<1.
	\end{equation}
	It is noteworthy that most of the procedures of Bayesian inference, which have been developed with the notion of the symmetric loss function, associating equal weight to the losses, either it is overestimation or underestimation, which is not realistic in various situations. Practically, we see a positive loss in most situations may be more severe than negative loss. To overcome the drawbacks of the symmetric loss function, asymmetric loss functions have been introduced. In this article, we consider symmetric and asymmetric loss functions to show the flexibility and applicability in various real-life situations of the calculated results. \\
	The mathematical form of these loss functions and their respective Bayes estimators are given as:
	The squared error loss function (SELF) is defined as
	\begin{equation}
		L_{1}(\delta, \lambda)=	(\delta- \lambda)^2,~~ \lambda>0.
	\end{equation}
	The Bayes estimator under SEL function is posterior mean $\left(\delta_{SEL}\right).$ 
	The LINEX loss function is defined as
	\begin{equation}\label{linexloss}
		L_{2}(\delta, \lambda)=e^{c(\delta- \lambda)}-c(\delta- \lambda)-1,\qquad c\ne 0
	\end{equation}
	with corresponding Bayes estimator as
	\begin{equation*}\label{linexbayes}
		\delta_{LINEX}=-\dfrac{1}{c}\ln \left(E(e^{-c \lambda})\right).
	\end{equation*}
	The General Entropy (GE) loss function is given as
	\begin{equation}\label{GEloss}
		L_3(\delta,\lambda)\propto \left(\dfrac{\delta}{\lambda}\right)^c-c\ln\left(\dfrac{\delta}{\lambda}\right)-1,\qquad c\ne  0
	\end{equation}
	with corresponding Bayes estimator as
	\begin{equation*}\label{GEbayes}
		\delta_{GE}=	\left[E(\lambda^{-c})\right]^{-1/c}.
	\end{equation*}
	Further we assume that $\theta$ has an informative prior with two parameter gamma distribution and is defined as
	\begin{equation}\label{prior.i}
		\pi(\theta)=\frac{b^{a}}{\Gamma{a}}{\theta^{a-1}}e^{-b\theta}, a,b>0.
	\end{equation}
	Now, the joint posterior density of $\theta$ is obtained by using \eqref{lik.eq} and \eqref{prior.i}, and is given as
	\begin{eqnarray}
		\pi(\theta|\boldsymbol{x})\propto  k\theta^{a}2^{(m+1)(n-1)+k}\left(\prod^{n-1}_{j=1}\gamma_j\right) \left(\prod^{n}_{i=1}\left(\theta+(2-\theta)e^{-x_{i}}\right)^{-(m+2)(k+1)}\right)\nonumber\\
		\times exp{\left[-\sum_{i=1}^{n-1}x_{i}(m+1)-kx_{n}-b\theta\right]},\,\,\, 0<\theta<1.
	\end{eqnarray} 
	
	It is observed that exact solution of posterior expectation is cumbersome due to the complex nature of posterior distribution. Hence, explicit form of Bayes estimators is difficult to obtain. We require some approximation tools for the computation purpose. In this direction two well known methods Lindley Approximation and Markov Chain Monte Carlo methods are considered. \par
	\subsection{Lindley's Approximation}
	\cite{lindley1980approximate} proposed a method of approximation of the ratio of the integrals and developed numerical results. 
	For our case, 
	\begin{equation}\label{lind.posterior}
		E(w(\theta)|\boldsymbol{x})=\dfrac{\int  w(\theta)e^{\mathbb{L}( \theta)+\rho( \theta)}d\theta}{\int e^{\mathbb{L}(\theta)+\rho( \theta)}d\theta}
	\end{equation}
	where $w(\theta)$ is the function of $\theta$, $\mathbb{L}( \theta)=\ln  L(\theta|\boldsymbol{x}) $ and $\rho( \theta)=\ln  \pi(\theta)$
	Using the Lindley approximation method in our case, $E(w(\theta)|\boldsymbol{x})$ can be approximated as
	\begin{align}\label{lind.bayesest}
		E(w( \hat{\theta})|\boldsymbol{x})&\approx w(\hat{\theta})+\dfrac{1}{2}\left[\left(\hat{w}_{\theta\theta}+2\hat{w}_{\theta}\hat{p}_{\theta}\right)\hat{\sigma}_{\theta\theta}\right]+\dfrac{1}{2}\left[\left(\hat{L}_{\theta\theta\theta}\hat{\sigma}_{\theta\theta\theta}\right)\hat{w}_{\theta}\hat{\sigma}_{\theta\theta}\right],
	\end{align}
	where 
	\begin{align}\label{lindley.1}
		\left.\begin{array}{ll}
			\hat{\theta}= \text{MLE of}~ \theta,~\hat{w_{\theta}}=\dfrac{\partial\hat{ w(\theta)}}{\partial\hat{\theta}},	\hat{w_{\theta\theta}}=\dfrac{\partial^2 w(\hat{\theta})}{\partial \hat{\theta}^2}\\ \hat{\sigma_{\theta\theta}}= -\dfrac{1}{\mathbb{\hat{L}}_{\theta\theta}},\,\
			\mathbb{\hat{L}}_{\theta\theta}= \dfrac{\partial^2 \ln  L(\theta|\boldsymbol{x})}
			{\partial \hat{\theta}^2},\,\,
			\mathbb{\hat{L}}_{\theta\theta\theta}= \dfrac{\partial^3 \ln  L(\hat{\theta}|\boldsymbol{x}) }{\partial \hat{\theta}^3},\,\
			\hat{p}_{{\theta}}=\dfrac{\partial \rho(\hat{\theta})}{\partial \hat{\theta}}.
		\end{array}\right\rbrace
	\end{align}
	
	Now, concerning our problem, the quantities given in \eqref{lindley.1} are reduced to the following forms:
	\begin{align}
		\left.
		\begin{array}{ll}
			L_{\theta \theta}= (m+2)\sum_{i=1}^{n-1}\dfrac{(1-e^{-x_{i}})}{\left(\theta+(2-\theta)e^{-x_{i}}\right)^{2}}-\dfrac{n}{\theta^{2}}+\dfrac{(k+1)(1-e^{-x_{n}})^2}{\left(\theta+(2-\theta)e^{-x_{i}}\right)^{2}}\\
			L_{\theta\theta\theta}= -2(m+2)\sum_{i=1}^{n-1}\dfrac{(1-e^{-x_{i}})^{3}}{\left(\theta+(2-\theta)e^{-x_{i}}\right)^{3}}+\dfrac{2n}{\theta^{3}}-\dfrac{(k+1)(1-e^{-x_{n}})^3}{\left(\theta+(2-\theta)e^{-x_{i}}\right)^{3}}\\
			\hat{p}_{{\theta}}= \dfrac{a(1-\theta)-1}{\theta}.
		\end{array}\right\rbrace
	\end{align}
	Further, we use the derived quantities to achieve the Bayes estimators under different loss functions.  Besides $w(\theta)$ and its derivatives, the quantities are common in each form of Bayes estimators. Under SELF, the posterior mean is the Bayes estimator, so for Bayes estimator of $\theta$ under SELF,
	\begin{equation*}
		w(\theta)=\theta ,~w_{\theta}=1,w_{\theta\theta}=0.
	\end{equation*}
	Under LINEX loss function, the quantity $w(\theta)$ and its derivatives, 
	\begin{equation*}
		w(\theta)=e^{-c\theta},~w_{\theta}=-ce^{-c\theta},~w_{\theta\theta}=c^2e^{-c\theta}.
	\end{equation*}
	Under GE loss function, the quantity $w(\theta)$ and its derivatives, 
	\begin{equation*}
		w(\theta)=\theta^{-c},~w_{\theta}=-c\theta^{-c-1},~w_{\theta\theta}=c(c+1)\theta^{-c-2}
	\end{equation*}
	Using the above quantities of $w(\theta)$ and it's derivative, Bayes estimator of $\theta$ can be obtained under different loss functions. 
	
	\subsection{Markov Chain Monte Carlo}
	In this section, we use the MCMC technique to obtain the approximate Bayes estimator of $\theta$ under the loss functions considered in this article. The MCMC technique is used to generate the random sample of unknown quantities with the help of posterior densities. The generated sample is then used to obtain the Bayes estimator with respect to the loss functions. It can be noticed that the posterior density of $\theta$ can not be reduced analytically to any known distribution. So, it is not easy to generate  the sample of $\theta$. We utilize the Metropolis Hasting (MH) algorithm with proposal density as a normal distribution (see \cite{gelman2013bayesian}, \cite{arshad2022stress}, \cite{azhad2022statistical}) to generate a sample. The following algorithm is used for generation of sample:
	\begin{description}
		\item[(i)] Initiate with prefixed value of $\theta$ as $\theta^{1}.$
		\item[(ii)] For $j=2$, generate $\theta^{j}$ by MH algorithm with the aid of normal distribution as proposal density from $\pi(\theta^{j-1}|\boldsymbol{x})$.
		\item[(iii)] Repeat (ii) for $j=3,\ldots \ldots,T$ times and obtain $(\theta^1,\theta^2,\ldots,\theta^T).$
	\end{description}
	Now, after employing algorithm and obtaining samples of $\theta,$ we calculate the Bayes estimators of $\theta$ as, $\delta^\theta_{SEL}=\dfrac{1}{T}\sum_{i=1}^{T}\theta^i$, $\delta^\theta_{LINEX}=-\dfrac{1}{c} \ln \left(\dfrac{1}{T}\sum_{i=1}^{T}e^{-c\theta^i}\right)$ and $\delta^\theta_{GE}=\left(\dfrac{1}{T}\sum_{i=1}^{T}(\theta^i)^{-c}\right)^{-1/c}$ under SELF, LINEX and GE loss functions, respectively.
	
	\section{Simulation Study}\label{sec:simu}
	In this section, a simulation study is conducted to monitor the performances of various derived Bayes estimators. The Bayes estimators are derived under the setup of \textit{gos}. As, \textit{gos} provides a unified approach for several ordered random variable based models. So, to show the practical applicability of derived results, we consider order statistics as a sub-model of \textit{gos}. We see that for $m=0$ and $k=1$ the \textit{gos} model to reduces to order statistics. The performances of Bayes estimators based on order statistics are measured using the criteria of average estimate (AE), average bias (AB) and mean squared error (MSE). For order statistics, the AEs, ABs and respective MSEs are reported in Tables [\ref{os.sel}-\ref{os.ge}] for various configurations i.e., $\theta=(0.3,0.6)$, $n\in(10,20,30,40,50)$, Prior I $(a=2,b=1)$ and Prior II $(a=2,b=2)$. These results are obtained by replicating the process 1000 times.
	From Table [\ref{os.sel}-\ref{os.ge}], the following observations are made:
	\begin{itemize}
		\item From Table [\ref{os.sel}], it is observed that under SELF, Bayes estimators of Lindley approximation are showing smaller MSEs than MCMC estimates. It is also observed that estimators coming from Prior II are showing smaller MSEs than estimators based on Prior I.
		\item From Table [\ref{os.linex}], it is observed that under LINEX loss, Bayes estimators of Lindley approximation are showing smaller MSEs than MCMC estimates. It is also observed that estimators coming from Prior II are showing smaller MSEs than estimators based on Prior I.
		\item From Table [\ref{os.ge}], it is observed that under GE loss, Bayes estimators of Lindley approximation are showing smaller MSEs than MCMC estimates. It is also observed that estimators coming from Prior II are showing smaller MSEs than estimators based on Prior I.
	\end{itemize}
	From above observations, it is clear that in our particular study Bayes estimators based on Lindley approximation are performing better than MCMC estimators. One can also draw the inference that as we increase the value of $b$ in simulation study then estimators` performance are improving. In general it is found that the estimators based on GE loss function seems to be performing better for $c\in (-0.5,0.5,1)$. For the behaviour of generated chains of MCMC, trace plots of $\theta$ are given in Figure [\ref{plots}]. All these results are discussed for this particular study.

	\begin{table}[htbp]
		\centering
		\caption{Lindley and MCMC Bayes estimates of unknown quantities based on Order Statistics under SEL function}
		\begin{tabular}{cc|ccc|ccc}
			\toprule
			\multicolumn{8}{c}{Prior I} \\
			\midrule
			\multirow{2}[4]{*}{$\theta$} & \multirow{2}[4]{*}{$n$} & \multicolumn{3}{c|}{Lindley} & \multicolumn{3}{c}{MCMC} \\
			\cmidrule{3-8}           &       & \multicolumn{1}{c|}{AE} & \multicolumn{1}{c|}{AB} & MSE   & \multicolumn{1}{c|}{AE} & \multicolumn{1}{c|}{AB} & \multicolumn{1}{c|}{MSE} \\
			\midrule
			0.3   & 10    & 0.384991 & 0.084991 & 0.015136 & \multicolumn{1}{r}{0.510198} & \multicolumn{1}{r}{0.210198} & \multicolumn{1}{r}{0.066864} \\
			0.3   & 20    & 0.336058 & 0.036058 & 0.004984 & \multicolumn{1}{r}{0.440272} & \multicolumn{1}{r}{0.140272} & \multicolumn{1}{r}{0.040311} \\
			0.3   & 30    & 0.319147 & 0.019147 & 0.002806 & \multicolumn{1}{r}{0.401055} & \multicolumn{1}{r}{0.101055} & \multicolumn{1}{r}{0.026447} \\
			0.3   & 40    & 0.310965 & 0.010965 & 0.001968 & \multicolumn{1}{r}{0.377314} & \multicolumn{1}{r}{0.077314} & \multicolumn{1}{r}{0.018298} \\
			0.3   & 50    & 0.306136 & 0.006136 & 0.001543 & \multicolumn{1}{r}{0.362261} & \multicolumn{1}{r}{0.062261} & \multicolumn{1}{r}{0.013431} \\
			0.6   & 10    & 0.737414 & 0.137414 & 0.046002 & \multicolumn{1}{r}{0.673480} & \multicolumn{1}{r}{0.073480} & \multicolumn{1}{r}{0.018968} \\
			0.6   & 20    & 0.653733 & 0.053733 & 0.016267 & \multicolumn{1}{r}{0.674178} & \multicolumn{1}{r}{0.074178} & \multicolumn{1}{r}{0.020668} \\
			0.6   & 30    & 0.625327 & 0.025327 & 0.009748 & \multicolumn{1}{r}{0.670912} & \multicolumn{1}{r}{0.070912} & \multicolumn{1}{r}{0.020700} \\
			0.6   & 40    & 0.611830 & 0.011830 & 0.007142 & \multicolumn{1}{r}{0.668367} & \multicolumn{1}{r}{0.068367} & \multicolumn{1}{r}{0.019633} \\
			0.6   & 50    & 0.603981 & 0.003981 & 0.005776 & \multicolumn{1}{r}{0.665214} & \multicolumn{1}{r}{0.065214} & \multicolumn{1}{r}{0.018427} \\
			\midrule
			\multicolumn{8}{c}{Prior II} \\
			\midrule
			\multirow{2}[4]{*}{$\theta$} & \multirow{2}[4]{*}{$n$} & \multicolumn{3}{c|}{Lindley} & \multicolumn{3}{c}{MCMC} \\
			\cmidrule{3-8}           &       & \multicolumn{1}{c|}{AE} & \multicolumn{1}{c|}{AB} & MSE   & \multicolumn{1}{c|}{AE} & \multicolumn{1}{c|}{AB} & \multicolumn{1}{c|}{MSE} \\
			\midrule
			0.3   & 10    & 0.368686 & 0.068686 & 0.011495 & 0.425618 & 0.125618 & 0.030661 \\
			0.3   & 20    & 0.326848 & 0.026848 & 0.004064 & 0.391349 & 0.091349 & 0.021955 \\
			0.3   & 30    & 0.312646 & 0.012646 & 0.002435 & 0.370522 & 0.070522 & 0.016353 \\
			0.3   & 40    & 0.305899 & 0.005899 & 0.001785 & 0.356750 & 0.056750 & 0.012406 \\
			0.3   & 50    & 0.301975 & 0.001975 & 0.001443 & 0.347253 & 0.047253 & 0.009743 \\
			0.6   & 10    & 0.672183 & 0.072183 & 0.024382 & 0.564762 & -0.035238 & 0.012346 \\
			0.6   & 20    & 0.616891 & 0.016891 & 0.011165 & 0.586416 & -0.013585 & 0.011931 \\
			0.6   & 30    & 0.599317 & -0.000683 & 0.007873 & 0.596816 & -0.003184 & 0.011744 \\
			0.6   & 40    & 0.591564 & -0.008436 & 0.006325 & 0.603181 & 0.003181 & 0.011081 \\
			0.6   & 50    & 0.587333 & -0.012667 & 0.005413 & 0.607671 & 0.007670 & 0.010450 \\
			\bottomrule
		\end{tabular}%
		\label{os.sel}%
	\end{table}%
	\begin{landscape}
		\begin{table}[htbp]
			\centering
			\caption{Lindley and MCMC Bayes estimates of unknown quantities based on Order Statistics under LINEX loss function}
			\begin{tabular}{cc|rrr|rrr|rrr|rrr}
				\toprule
				\multicolumn{8}{c|}{Prior I}                                  & \multicolumn{6}{c}{ Prior II} \\
				\midrule
				\multicolumn{14}{c}{                                      c=0.5} \\
				\midrule
				\multirow{2}[4]{*}{$\theta$} & \multirow{2}[4]{*}{$n$} & \multicolumn{3}{c|}{Lindley} & \multicolumn{3}{c|}{MCMC} & \multicolumn{3}{c|}{Lindley} & \multicolumn{3}{c}{MCMC} \\
				\cmidrule{3-14}           &       & \multicolumn{1}{c}{AE} & \multicolumn{1}{c}{AB} & \multicolumn{1}{c|}{MSE} & \multicolumn{1}{c}{AE} & \multicolumn{1}{c}{AB} & \multicolumn{1}{c|}{MSE} & \multicolumn{1}{c}{AE} & \multicolumn{1}{c}{AB} & \multicolumn{1}{c|}{MSE} & \multicolumn{1}{c}{AE} & \multicolumn{1}{c}{AB} & \multicolumn{1}{c}{MSE} \\
				\midrule
				0.3   & 10    & 0.349994 & 0.049994 & 0.002500 & 0.501032 & 0.201032 & 0.062809 & 0.340709 & 0.040709 & 0.001657 & 0.419163 & 0.119163 & \multicolumn{1}{c}{0.028803} \\
				0.3   & 30    & 0.317992 & 0.017992 & 0.000324 & 0.397188 & 0.097188 & 0.025273 & 0.314757 & 0.014757 & 0.000218 & 0.367562 & 0.067562 & \multicolumn{1}{c}{0.015709} \\
				0.3   & 50    & 0.312768 & 0.012768 & 0.000164 & 0.360179 & 0.060179 & 0.012975 & 0.310574 & 0.010574 & 0.000112 & 0.345495 & 0.045495 & \multicolumn{1}{c}{0.009452} \\
				0.6   & 10    & 0.677567 & 0.077567 & 0.006017 & 0.664959 & 0.064959 & 0.018123 & 0.640172 & 0.040172 & 0.001614 & 0.557881 & -0.042119 & \multicolumn{1}{c}{0.012926} \\
				0.6   & 30    & 0.628020 & 0.028020 & 0.000786 & 0.665345 & 0.065345 & 0.020029 & 0.615045 & 0.015045 & 0.000227 & 0.592384 & -0.007616 & \multicolumn{1}{c}{0.011784} \\
				0.6   & 50    & 0.620128 & 0.020128 & 0.000407 & 0.660904 & 0.060904 & 0.017875 & 0.611332 & 0.011332 & 0.000130 & 0.604261 & 0.004261 & \multicolumn{1}{c}{0.010375} \\
				\midrule
				\multicolumn{14}{c}{                                      c=1} \\
				\midrule
				0.3   & 10    & 0.348252 & 0.048252 & 0.002328 & 0.492072 & 0.192072 & 0.058949 & 0.338771 & 0.038771 & 0.001503 & 0.412879 & 0.112879 & 0.027055 \\
				0.3   & 30    & 0.317257 & 0.017257 & 0.000298 & 0.393430 & 0.093430 & 0.024159 & 0.313998 & 0.013998 & 0.000196 & 0.364671 & 0.064671 & 0.015096 \\
				0.3   & 50    & 0.312258 & 0.012258 & 0.000151 & 0.358142 & 0.058142 & 0.012539 & 0.310052 & 0.010052 & 0.000102 & 0.343769 & 0.043768 & 0.009173 \\
				0.6   & 10    & 0.669357 & 0.069357 & 0.004811 & 0.656378 & 0.056378 & 0.017379 & 0.631175 & 0.031175 & 0.000972 & 0.551051 & -0.048949 & 0.013575 \\
				0.6   & 30    & 0.624925 & 0.024925 & 0.000622 & 0.659792 & 0.059792 & 0.019401 & 0.611850 & 0.011850 & 0.000141 & 0.587986 & -0.012014 & 0.011853 \\
				0.6   & 50    & 0.618008 & 0.018008 & 0.000326 & 0.656619 & 0.056619 & 0.017352 & 0.609161 & 0.009161 & 0.000085 & 0.600878 & 0.000878 & 0.010319 \\
				\midrule
				\multicolumn{14}{c}{                                      c=-0.5} \\
				\midrule
				0.3   & 10    & 0.353197 & 0.053197 & 0.002830 & 0.519553 & 0.219553 & 0.07111 & 0.344339 & 0.044339 & 0.001966 & 0.432240 & 0.132240 & 0.032632 \\
				0.3   & 30    & 0.319420 & 0.019420 & 0.000378 & 0.405032 & 0.105032 & 0.027683 & 0.316240 & 0.016240 & 0.000264 & 0.373553 & 0.073553 & 0.017031 \\
				0.3   & 50    & 0.313770 & 0.013770 & 0.000190 & 0.364387 & 0.064387 & 0.013908 & 0.311602 & 0.011602 & 0.000135 & 0.349043 & 0.049043 & 0.010048 \\
				0.6   & 10    & 0.692079 & 0.092079 & 0.008479 & 0.681917 & 0.081917 & 0.01991 & 0.657100 & 0.057100 & 0.003261 & 0.571684 & -0.028316 & 0.011836 \\
				0.6   & 30    & 0.633950 & 0.033950 & 0.001154 & 0.676483 & 0.076483 & 0.021413 & 0.621291 & 0.021291 & 0.000454 & 0.601281 & 0.001281 & 0.011734 \\
				0.6   & 50    & 0.624240 & 0.024240 & 0.000590 & 0.669543 & 0.069543 & 0.019007 & 0.615601 & 0.015601 & 0.000245 & 0.611105 & 0.011105 & 0.010544 \\
				\bottomrule
			\end{tabular}%
			\label{os.linex}%
		\end{table}%
		
	\end{landscape}	
	
	\begin{landscape}
		\begin{table}[htbp]
			\centering
			\caption{Lindley and MCMC Bayes estimates of unknown quantities based on Order Statistics under GE loss function}
			\begin{tabular}{cc|rrr|rrr|rcc|ccc}
				\toprule
				\multicolumn{8}{c|}{Prior I}                                  & \multicolumn{6}{c}{ Prior II} \\
				\midrule
				\multicolumn{14}{c}{                                      c=0.5} \\
				\midrule
				\multirow{2}[4]{*}{$\theta$} & \multirow{2}[4]{*}{$n$} & \multicolumn{3}{c|}{Lindley} & \multicolumn{3}{c|}{MCMC} & \multicolumn{3}{c|}{Lindley} & \multicolumn{3}{c}{MCMC} \\
				\cmidrule{3-14}           &       & \multicolumn{1}{l}{AE} & \multicolumn{1}{l}{AB} & \multicolumn{1}{l|}{MSE} & \multicolumn{1}{l}{AE} & \multicolumn{1}{l}{AB} & \multicolumn{1}{l|}{MSE} & \multicolumn{1}{l}{AE} & \multicolumn{1}{l}{AB} & \multicolumn{1}{l|}{MSE} & \multicolumn{1}{l}{AE} & \multicolumn{1}{l}{AB} & \multicolumn{1}{l}{MSE} \\
				\midrule
				0.3   & 10    & 0.331193 & 0.031193 & 0.000973 & 0.450692 & 0.150692 & 0.046511 & \multicolumn{1}{c}{0.320908} & 0.020908 & 0.000437 & \multicolumn{1}{r}{0.377885} & \multicolumn{1}{r}{0.077885} & 0.021036 \\
				0.3   & 30    & 0.310986 & 0.010986 & 0.000121 & 0.374346 & 0.074346 & 0.020539 & \multicolumn{1}{c}{0.307626} & 0.007626 & 0.000058 & \multicolumn{1}{r}{0.348003} & \multicolumn{1}{r}{0.048003} & 0.012972 \\
				0.3   & 50    & 0.307977 & 0.007977 & 0.000064 & 0.346549 & 0.046549 & 0.011006 & \multicolumn{1}{c}{0.305720} & 0.005719 & 0.000033 & \multicolumn{1}{r}{0.333101} & \multicolumn{1}{r}{0.033101} & 0.008145 \\
				0.6   & 10    & 0.641840 & 0.041840 & 0.001751 & 0.626669 & 0.026669 & 0.017829 & \multicolumn{1}{c}{0.603471} & 0.003471 & 0.000012 & \multicolumn{1}{r}{0.522843} & \multicolumn{1}{r}{-0.077157} & 0.018722 \\
				0.6   & 30    & 0.615275 & 0.015275 & 0.000234 & 0.643388 & 0.043388 & 0.018879 & \multicolumn{1}{c}{0.602157} & 0.002157 & 0.000005 & \multicolumn{1}{r}{0.573007} & \multicolumn{1}{r}{-0.026993} & 0.013124 \\
				0.6   & 50    & 0.611463 & 0.011463 & 0.000133 & 0.644665 & 0.044665 & 0.016714 & \multicolumn{1}{c}{0.602582} & 0.002581 & 0.000007 & \multicolumn{1}{r}{0.590195} & \multicolumn{1}{r}{-0.009805} & 0.010757 \\
				\midrule
				\multicolumn{14}{c}{                                      c=1} \\
				\midrule
				0.3   & 10    & 0.323031 & 0.023031 & 0.000531 & 0.430095 & 0.130095 & 0.040656 & 0.312836 & 0.012836 & 0.000165 & 0.361615 & 0.061615 & 0.018591 \\
				0.3   & 30    & 0.308239 & 0.008239 & 0.000068 & 0.365679 & 0.065679 & 0.018885 & 0.304887 & 0.004887 & 0.000024 & 0.340646 & 0.040645 & 0.012062 \\
				0.3   & 50    & 0.306126 & 0.006126 & 0.000038 & 0.341443 & 0.041443 & 0.010324 & 0.303870 & 0.003870 & 0.000015 & 0.328473 & 0.028473 & 0.007707 \\
				0.6   & 10    & 0.625695 & 0.025695 & 0.000660 & 0.609032 & 0.009032 & 0.018406 & 0.588545 & -0.011455 & 0.000131 & 0.507690 & -0.092310 & 0.021798 \\
				0.6   & 30    & 0.609798 & 0.009798 & 0.000096 & 0.633769 & 0.033769 & 0.018522 & 0.596818 & -0.003182 & 0.000010 & 0.564812 & -0.035188 & 0.013821 \\
				0.6   & 50    & 0.607764 & 0.007764 & 0.000061 & 0.637646 & 0.037646 & 0.016280 & 0.598938 & -0.001062 & 0.000001 & 0.584265 & -0.015735 & 0.010979 \\
				\midrule
				\multicolumn{14}{c}{                                      c=-0.5} \\
				\midrule
				0.3   & 10    & 0.345700 & 0.045700 & 0.002089 & 0.49088 & 0.19088 & 0.059691 & 0.336025 & 0.036024 & 0.001298 & 0.409961 & 0.109961 & 0.027099 \\
				0.3   & 30    & 0.316254 & 0.016254 & 0.000265 & 0.392045 & 0.092045 & 0.02432 & 0.312968 & 0.012968 & 0.000168 & 0.362949 & 0.062949 & 0.015117 \\
				0.3   & 50    & 0.311565 & 0.011565 & 0.000134 & 0.356959 & 0.056959 & 0.012558 & 0.309346 & 0.009346 & 0.000088 & 0.342492 & 0.042492 & 0.009162 \\
				0.6   & 10    & 0.672074 & 0.072074 & 0.005195 & 0.658899 & 0.058899 & 0.018183 & 0.634192 & 0.034192 & 0.001169 & 0.551398 & -0.048602 & 0.014032 \\
				0.6   & 30    & 0.625961 & 0.025960 & 0.000675 & 0.661973 & 0.061973 & 0.019978 & 0.612918 & 0.012918 & 0.000167 & 0.589019 & -0.010981 & 0.012090 \\
				0.6   & 50    & 0.618717 & 0.018717 & 0.000352 & 0.658458 & 0.058458 & 0.017791 & 0.609886 & 0.009886 & 0.000099 & 0.601902 & 0.001902 & 0.010494 \\
				\bottomrule
			\end{tabular}%
			\label{os.ge}%
		\end{table}%
		
	\end{landscape}	
	\begin{figure}[h]
		\centering
		\begin{minipage}{.5\textwidth}
			\centering
			\includegraphics[width=1\linewidth]{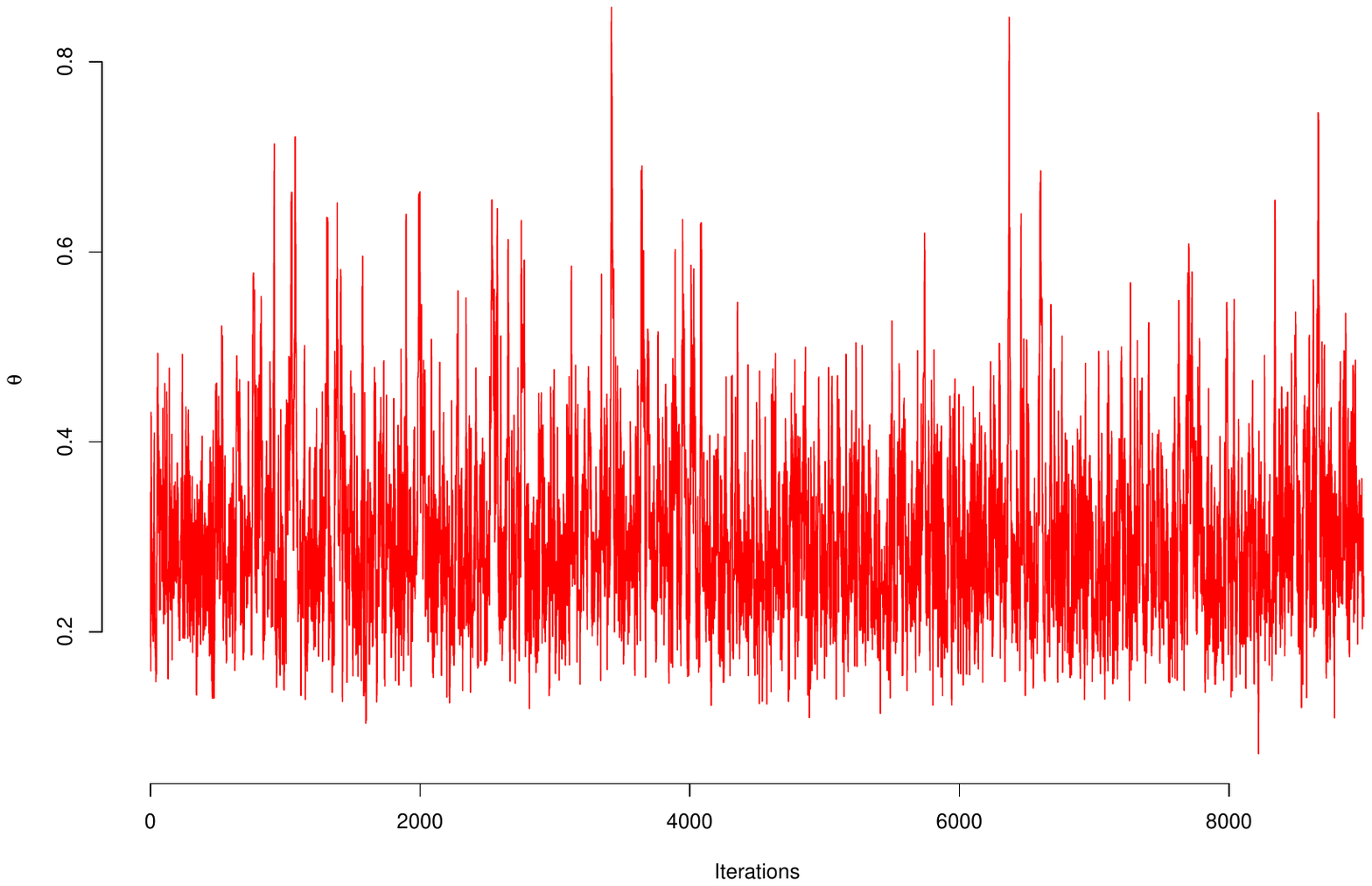}
			\captionof*{figure}{For $\theta=0.3$, $a=2,~b=1$}
			\label{fig:t2}
		\end{minipage}%
		\begin{minipage}{.5\textwidth}
			\centering
			\includegraphics[width=1\linewidth]{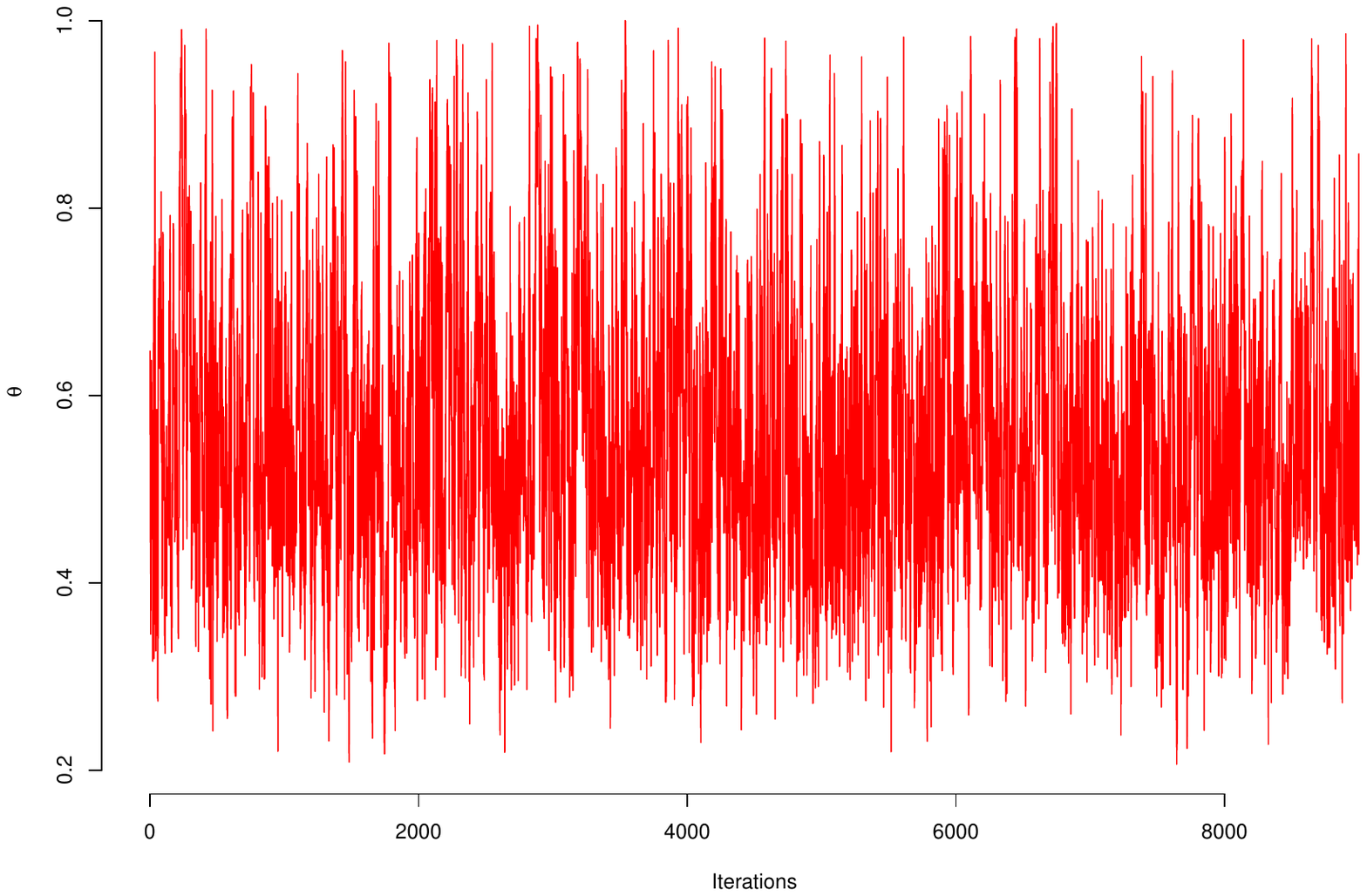}
			\captionof*{figure}{For $\theta=0.5$, $a=2,~b=1$}
			\label{fig:t1}
		\end{minipage}
		\caption{Trace plots of $\theta$}
		\label{plots}
	\end{figure}
	
	\section{Data Analysis}\label{sec:real}
	\begin{application} \label{appli1}
		In order to exhibit the application of proposed results in the field of demographic study, we have considered mortality rates of 30 days of Netherlands country due to COVID-19 from March 31, 2020 to April 30, 2020. The COVID-19 is the third-highest cause of deaths in $2020$ which has been revealed by the US Centers for Disease Control and Prevention (CDC). The mortality rate is calculated by the ratio of number of deaths and total number of cases (reported cases per 100,000) (see \cite{albalawi2022estimation}). The mortality rates of 30 days of Netherlands is reported in the Table [\ref{tab:realdata1}].
	\end{application}
	\begin{table}[h]
		\centering
		\caption{Mortality rates of Netherlands}
		\begin{tabular}{ccccccccccc}
			\toprule
			& 14.918 & 10.656 & 12.274 & 10.289 & 10.832 & 7.099 & 5.928 & 13.211 & 7.968 & 7.584\\
			& 5.555 & 6.027 & 4.097 & 3.611 & 4.960 & 7.498 & 6.940 & 5.307 & 5.048 & 2.857\\
			& 2.254 & 5.431 & 4.462 & 3.883 & 3.461 & 3.647 & 1.974 & 1.273 & 1.416 & 4.235\\
			\bottomrule
		\end{tabular} 
		\label{tab:realdata1}
	\end{table} 
	We tend towards the Kolmogorov - Smirnov (K-S) test to see whether the data supports HLG distribution or not? After employing KS test, it is observed that the data given in Table [\ref{tab:realdata1}] support HLG distribution for $\theta = 0.0079$ with $p-$ value $0.1915$ and KS distance  $0.1943$. This claim is also supported with the aid of Figure [\ref{fig:fittingmortality}].
	\begin{figure}
		\centering
		\includegraphics[width=0.7\linewidth]{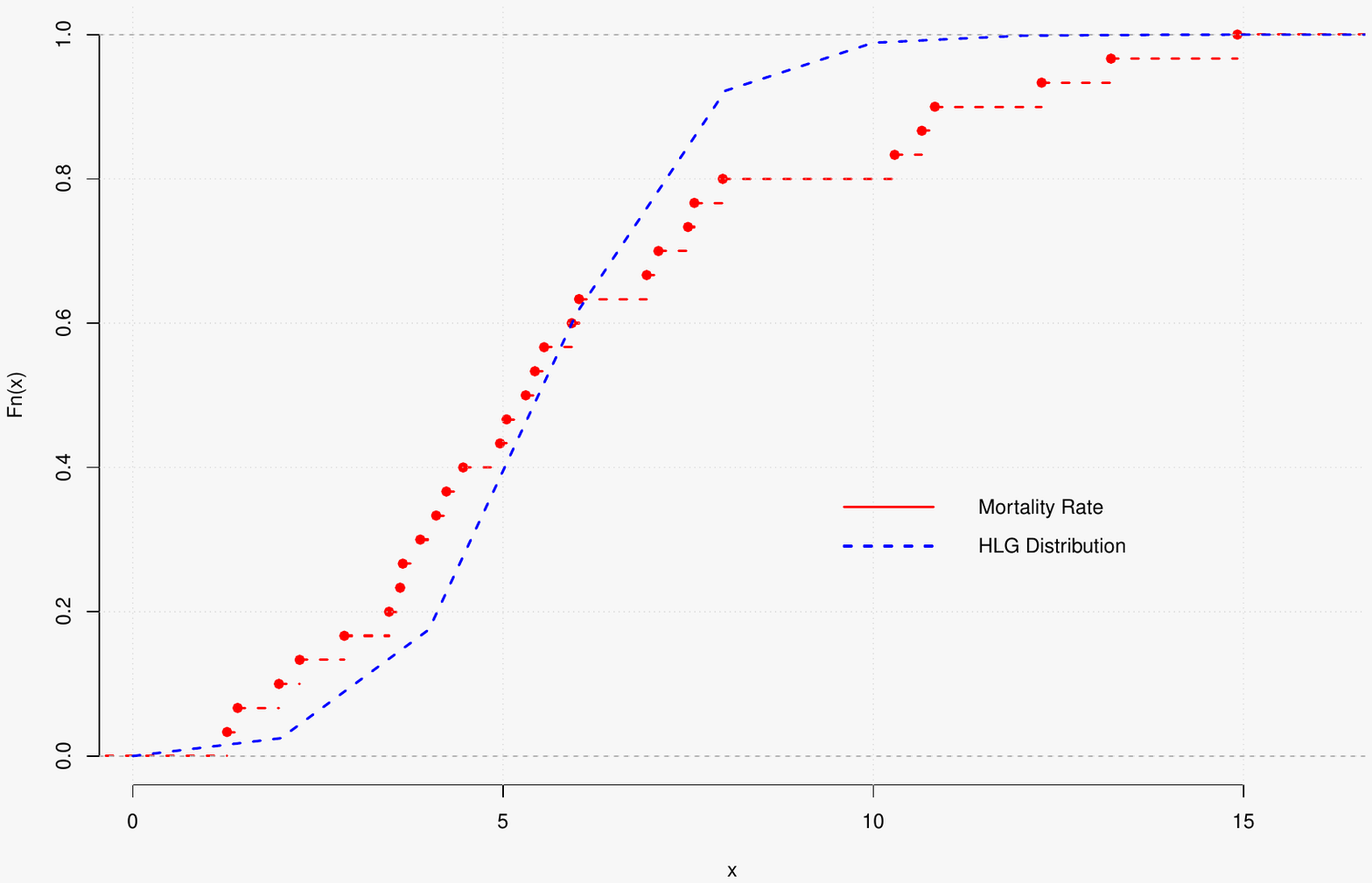}
		\caption{Fitting of Mortality Data with HLG Distribution}
		\label{fig:fittingmortality}
	\end{figure}
	From the $p-$ value, we observe that the mortality data set of Netherlands can be analyzed by using HLG distribution. In this direction, we have arranged the data in increasing order and estimated $\theta$ for classical and Bayesian scheme and the results are reported in Table [\ref{tab:est_mortality}] under the setup of order statistics. From the simulation study, it is concluded that estimators based on GE loss function are performing better. So, for descriptive analysis of the mortality data we may consider $\delta^\theta_{GE}=0.01111$ for $a=2,~b=2$ and $c=1.5.$ 
	\begin{table}[htbp]
		\centering
		\caption{Bayes estimates of $\theta$ for mortality data of Netherlands}
		\begin{tabular}{c|c|c|ccc|ccc}
			\toprule
			\multirow{2}[3]{*}{$c$} & \multirow{2}[3]{*}{$a$} & \multirow{2}[3]{*}{$b$} & \multicolumn{3}{c|}{Lindley} & \multicolumn{3}{c}{MCMC} \\
			\cmidrule{4-9}          &       &       & $\delta^\theta_{SEL}$ & $\delta^\theta_{LIN}$ & \multicolumn{1}{c|}{$\delta^\theta_{GE}$} & $\delta^\theta_{SEL}$ & $\delta^\theta_{LIN}$ & $\delta^\theta_{GE}$ \\
			\midrule
			\multirow{3}[1]{*}{-0.5} & \multirow{12}[2]{*}{2} & 1     & 0.01081 & 0.01081 & 0.01032 & 0.01200 & 0.01201 & 0.01158 \\
			&       & 1.5   & 0.01066 & 0.01066 & 0.01031 & 0.01135 & 0.01136 & 0.01094 \\
			&       & 2     & 0.01054 & 0.01054 & 0.01030 & 0.01168 & 0.01169 & 0.01126 \\
			\multirow{3}[0]{*}{0.5} &       & 1     & 0.01046 & 0.01046 & 0.01138 & 0.01105 & 0.01105 & 0.00996 \\
			&       & 1.5   & 0.01033 & 0.01032 & 0.01120 & 0.01117 & 0.01116 & 0.01021 \\
			&       & 2     & 0.01019 & 0.01019 & 0.01103 & 0.01113 & 0.01113 & 0.01012 \\
			\multirow{3}[0]{*}{1} &       & 1     & 0.01007 & 0.01007 & 0.01137 & 0.01139 & 0.01138 & 0.00992 \\
			&       & 1.5   & 0.00994 & 0.00994 & 0.01119 & 0.01139 & 0.01138 & 0.00984 \\
			&       & 2     & 0.00981 & 0.00981 & 0.01102 & 0.01142 & 0.01141 & 0.00970 \\
			\multirow{3}[1]{*}{1.5} &       & 1     & 0.00970 & 0.00970 & 0.01149 & 0.01065 & 0.01064 & 0.00929 \\
			&       & 1.5   & 0.00958 & 0.00958 & 0.01130 & 0.01203 & 0.01201 & 0.00989 \\
			&       & 2     & 0.00947 & 0.00946 & 0.01111 & 0.01108 & 0.01107 & 0.00935 \\
			\bottomrule
		\end{tabular}%
		\label{tab:est_mortality}%
	\end{table}%

	\begin{application}\label{real2}
		In order to exhibit the application of proposed study in the field of reliability, we have considered a lifetime data of traction motors.
		\cite{jung2007analysis} investigated the eligibility of warranty claim of the same dataset  by assuming that a two dimensional warranty has been provided by the manufacture. The data set is taken from the maintenance records of a type of locomotive traction motor. The data present the time since inception of service and miles accumulated by different
		traction motors when they were returned to the maintenance depot upon failing. The original data set can be accessed from \cite{eliashberg1997calculating}. The considered data set is of bivariate nature containing age and usage of motors.  We have extracted only the age factor of the data set and performed analysis for the considered problem. The data set is given in Table [\ref{tab:realdata2}]
		
		\begin{table}[htbp]
			\centering
			\caption{Age for traction motors}
			
			\begin{tabular}{cccccccccc}
				\toprule
				1.66  & 3.35  & 1.28  & 0.01  & 0.41  & 4.98  & 0.22  & 0.02  & 1.9   & 1.7 \\
				0.35  & 1.64  & 0.31  & 0.27  & 0.59  & 5.71  & 2.61  & 2.09  & 0.27  & 1.4 \\
				2.49  & 1.45  & 0.65  & 2.95  & 0.75  & 4.99  & 0.32  & 0.29  & 2.21  & 1.4 \\
				2.23  & 3.4   & 3.97  & 1.66  & 9.52  & 1.6   & 0.48  & 12    & 3.16  & 8.27 \\
				\bottomrule
			\end{tabular}%
			\label{tab:realdata2}%
		\end{table}%
	\end{application}
	Again using the same approach, we find that dataset given in Table [\ref{tab:realdata2}] support HLG distribution for $\theta=0.5010477$ with KS distance as $0.14129$ and $p-$ value as $0.4016.$  This claim is also supported with the aid of Figure [\ref{fig:realdata2}].
	\begin{figure}
		\centering
		\includegraphics[width=0.7\linewidth]{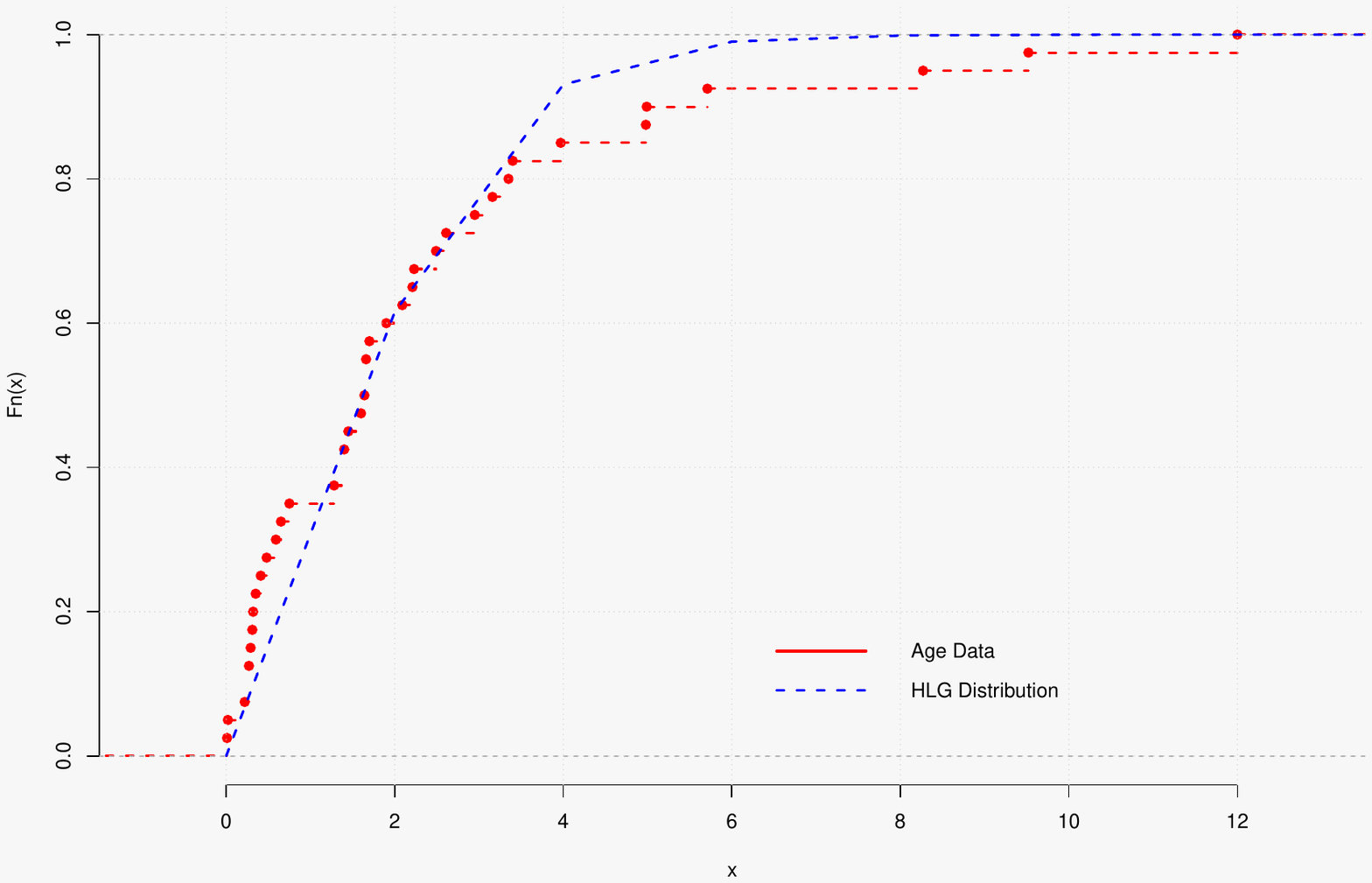}
		\caption{Fitting of  age for traction motors data with HLG Distribution}
		\label{fig:realdata2}
	\end{figure}
	From the $p-$ value, we observe that the age data for traction motors can be analyzed by using HLG distribution. In this direction, we have arranged the data in increasing order and estimated $\theta$ under classical and Bayesian scheme and the results are reported in Table [\ref{tab:age_data}] under the setup of order statistics. From the study, it is concluded that estimators based on GE loss function are performing better. So, for descriptive analysis of the mortality data we may consider $\delta^\theta_{GE}=0.57354$ for $a=2,~b=2$ and $c=1.5.$
	\begin{table}[htbp]
		\centering
		\caption{Bayes estimates of $\theta$ for age data of traction motors}
		\begin{tabular}{c|c|c|ccc|ccc}
			\toprule
			\multirow{2}[3]{*}{$c$} & \multirow{2}[3]{*}{$a$} & \multirow{2}[3]{*}{$b$} & \multicolumn{3}{c|}{Lindley} & \multicolumn{3}{c}{MCMC} \\
			\cmidrule{4-9}          &       &       & $\delta^\theta_{SEL}$ & $\delta^\theta_{LIN}$ & $\delta^\theta_{GE}$ & $\delta^\theta_{SEL}$ & $\delta^\theta_{LIN}$ & $\delta^\theta_{GE}$ \\ \midrule
			\multirow{3}[0]{*}{-0.5} & \multirow{12}[1]{*}{2} & 1     & 0.59329 & 0.59676 & 0.57232 & 0.59022 & 0.59604 & 0.58029 \\
			&       & 1.5   & 0.58152 & 0.58551 & 0.55980 & 0.55614 & 0.56229 & 0.54536 \\
			&       & 2     & 0.56976 & 0.57421 & 0.54743 & 0.54267 & 0.54803 & 0.53318 \\
			\multirow{3}[0]{*}{0.5} &       & 1     & 0.59310 & 0.58922 & 0.58511 & 0.58304 & 0.57769 & 0.55525 \\
			&       & 1.5   & 0.58130 & 0.57693 & 0.57048 & 0.57043 & 0.56420 & 0.53770 \\
			&       & 2     & 0.56952 & 0.56473 & 0.55640 & 0.56046 & 0.55457 & 0.52831 \\
			\multirow{3}[0]{*}{1} &       & 1     & 0.59280 & 0.58463 & 0.59700 & 0.58300 & 0.57004 & 0.53632 \\
			&       & 1.5   & 0.58105 & 0.57193 & 0.58071 & 0.57715 & 0.56739 & 0.54234 \\
			&       & 2     & 0.56931 & 0.55939 & 0.56529 & 0.54109 & 0.53097 & 0.50368 \\
			\multirow{3}[1]{*}{1.5} &       & 1     & 0.59260 & 0.57973 & 0.60928 & 0.58027 & 0.56259 & 0.52658 \\
			&       & 1.5   & 0.58083 & 0.56660 & 0.59073 & 0.57068 & 0.55197 & 0.51311 \\
			&       & 2     & 0.56908 & 0.55375 & 0.57354 & 0.56009 & 0.54345 & 0.50825 \\
			\bottomrule
		\end{tabular}%
		\label{tab:age_data}%
	\end{table}%
	\section{Conclusion}
	In this article, we have considered half logistic geometric (HLG) distribution under the setup of generalized order statistics (\textit{gos}). We have provided the expression for single and product moments of the considered distribution. Approximation techniques i.e., Lindley and MCMC, are employed for the computation of Bayes estimator of unknown quantity of HLG distribution. We have considered the case of order statistics to discuss the behaviour of derived estimators. It has been observed that the Bayes estimators based on Lindley approximation are performing better than MCMC estimators. Inference can be made that as we increase the value of hyperparameter $b$ , the performance of estimator is improving. In general, it is found that the estimators based on GE loss function seems to be performing better for $c\in (-0.5,0.5,1)$.
	For the application of the considered problem in the real world, data from two fields i.e., reliability and demography are taken and analyzed. \\
\section*{Declarations}

\begin{itemize}
	\item \textbf{Funding:} Not Applicable
	\item \textbf{Competing interests: }
	No potential competing interest was reported by the authors.
	\item \textbf{Ethics approval:} Not Applicable 
	\item \textbf{Consent to participate:} Not Applicable
	\item \textbf{Consent for publication:} 
	All authors provide consent for publication.
	\item \textbf{Availability of data and materials:} Not Applicable
	\item \textbf{Research involving Human Participants and/or Animals:} Not Applicable
	\item \textbf{Informed consent:} Not Applicable
	\item \textbf{Authors' contributions:} 
	Conceptualization: Neetu Gupta, S. K. Neogy, Qazi J. Azhad, Bhagwati Devi; Methodology: Neetu Gupta, Qazi J. Azhad; Formal analysis and investigation: Neetu Gupta, Qazi J. Azhad, Bhagwati Devi; Writing - original draft preparation: Neetu Gupta, S. K. Neogy; Writing - review and editing: S. K. Neogy, Qazi J. Azhad, Bhagwati Devi;  Supervision: S. K. Neogy
\end{itemize}
	
	\bibliography{reference}
\end{document}